\documentclass[reprint,aip,jcp,amsmath,amssymb]{revtex4-2}
\usepackage{parskip}
\usepackage{graphicx}
\usepackage{bm}
\usepackage{algpseudocode}
\usepackage[hidelinks]{hyperref}
\usepackage{tikz}
\usetikzlibrary{positioning,fit,backgrounds,arrows.meta,calc,shapes.geometric}
\definecolor{sqdcl}{RGB}{31,94,160}
\definecolor{sqdqu}{RGB}{122,47,140}
\definecolor{sqdio}{RGB}{60,60,60}
\newcommand{\bra}[1]{\langle\mskip3mu#1\mskip3mu\vert}
\newcommand{\ket}[1]{\vert\mskip3mu#1\mskip3mu\rangle}
\newcommand{\braket}[2]{\langle\mskip3mu#1\mskip3mu\ket{#2}}
\newcommand{\supp}{\operatorname{supp}}
\newcommand{\Tr}{\operatorname{Tr}}
\newtheorem{proposition}{Proposition}
\newtheorem{remark}{Remark}
\newcounter{algorithm}
\renewcommand{\thealgorithm}{\arabic{algorithm}}
\begin{document}

\title{Extending Sample-based Quantum Diagonalization to Systems of Linear Equations}

\author{Taha Selim}
\email{taha@tahaselim.com}
\affiliation{MolKet SAS, Paris, France}

\date{\today}

\begin{abstract}
    Solving systems of linear equations, $A\mathbf{x}=\mathbf{b}$, is a ubiquitous task 
    in computational science, and efficient quantum algorithms for this problem could 
    have wide-ranging applications. The Harrow--Hassidim--Lloyd (HHL) algorithm requires
    phase estimation, controlled rotations, and post-selection, which remain impractical 
    on noisy intermediate-scale quantum hardware~\cite{harrow2009quantum, 
    preskill2018nisq}. We extend sample-based quantum diagonalization (SQD), developed 
    for quantum chemistry, to linear systems~\cite{robledomoreno2025chemistry}. Several 
    concurrent works have explored related approaches for linear systems and Krylov-based 
    methods. In particular, sample-based quantum linear solvers, which measure quantum 
    Krylov states and reconstruct the solution classically, have recently been proposed 
    together with energy-norm and sample-complexity guarantees~\cite{ahn2026riemannian,yu2025skqd,park2026sqls}. 
    In this work, we formulate the approach within the SQD framework and analyze its 
    cost in the Pauli-string encoding. Following the sample-based Krylov variant of the 
    method, the Hermitian matrix $A$ is decomposed in the Pauli basis, the Trotterized 
    time evolution of $\ket{b}$ is sampled at several times to select computational-basis 
    configurations, and the system is solved classically by a Galerkin projection onto 
    the selected subspace~\cite{yu2025skqd}. The method applies to Hermitian matrices with a right-hand side of small support. We prove that the subspace solution is exact when the sampled configurations contain the Krylov support of $\mathbf{b}$ and optimal in the energy norm for symmetric positive-definite matrices, so that noise-induced configurations cannot degrade it as long as they do not displace the ideal ones. We validate the pipeline with circuits on a noiseless simulator, a noisy simulator, and IBM Quantum hardware for a $4\times4$ example, and test it on a localized problem of dimension $N=1024$ using exact (ideal) sampling distributions rather than circuits: SQD selects few configurations, reaches a small relative error of $10^{-5}$, and its truncated subspaces coincide with the optimal choice, while readout errors enlarge the subspace without degrading the accuracy. The tests also 
    identify the limitations of the approach: the selected subspace coincides with the 
    one found by inexpensive classical Krylov and residual-based selection, banded matrices require exponentially many Pauli terms in the binary encoding, and Trotter errors populate spurious configurations. These results delineate the conditions under 
    which a quantum sampling step can be advantageous. This work was initially presented 
    at the Quantum AI Conference 2026, held at Indiana University Indianapolis, 
    Indianapolis, IN, USA, on 14--16 August 2026. This is a preliminary version of the 
    manuscript; additional tests, details, and updates will be included in subsequent 
    versions.
\end{abstract}

\maketitle

\section{Introduction}
\label{sec:intro}

\par Solving systems of linear equations efficiently is a fundamental problem in both classical and quantum computing~\cite{dervovic2018primer}. Several classical algorithms exist, but their cost typically scales polynomially with the dimension $N$.
The Harrow--Hassidim--Lloyd (HHL) algorithm, a well-known quantum algorithm for solving linear systems, prepares a state proportional to $A^{-1}\ket{b}$ with a cost logarithmic in the dimension $N$, but it relies on quantum phase estimation, eigenvalue-controlled ancilla rotations, and post-selection~\cite{harrow2009quantum}. Moreover, the cost can scale poorly with the condition number of $A$ and the desired accuracy, so that the resulting circuits are deep, in particular for applications that require high precision, such as expectation values of observables. Such circuits are difficult to execute reliably on current noisy intermediate-scale quantum (NISQ) devices, and before fault-tolerant quantum computers become available it is essential to develop algorithms that can be executed natively within the limitations of present hardware~\cite{preskill2018nisq}.
This motivates the search for alternative quantum algorithms that can achieve similar goals with shallower circuits and potentially lower overall cost.

A class of hybrid algorithms that is well suited to such hardware uses the quantum computer only to \emph{sample} computational-basis configurations, and solves the problem classically in the subspace spanned by the sampled configurations. Sample-based quantum diagonalization, which has recently gained attention, follows this strategy for electronic ground states, without the parameter optimization of variational eigensolvers and without the deep circuits of phase estimation; SQD has been applied to chemical systems beyond the reach of exact diagonalization~\cite{kanno2023qsci,peruzzo2014variational,robledomoreno2025chemistry}. In the sample-based Krylov variant (SKQD), the sampled states are time-evolved states $e^{-iHt_j}\ket{\psi_0}$, which endows the selected subspace with the convergence properties of Krylov methods~\cite{yu2025skqd,piccinelli:25}. Sample-based approaches have very recently been applied to linear systems as well. Ahn and Park measure quantum Krylov states and reconstruct the solution of positive-definite systems arising from elliptic partial differential equations by a Galerkin solve on the sampled basis states~\cite{ahn2026riemannian}. Preliminary results of the present work were presented at the Quantum AI 2026 Conference (Indianapolis, 14--16 August 2026)~\cite{selim2026qai}, independently of and before the appearance of Refs.~\onlinecite{ahn2026riemannian,park2026sqls}; the approaches are complementary.

In this work, we extend this approach to linear systems. We sample the time-evolved states $e^{-iAt_j}\ket{b}$ following the SKQD protocol and solve the linear system by a Galerkin projection onto the sampled configurations. Two differences from SQD for chemistry should be stated at the outset: the classical step is a linear solve rather than an eigenvalue problem, and no configuration recovery is required, since a general linear system has no particle-number or spin symmetry to restore. Throughout, $A$ is Hermitian (symmetric positive-definite for the optimality statements) and $b$ has small support, e.g.\ a single basis state; a non-Hermitian $A$ would require the Hermitian dilation, which is indefinite and is not treated here.

Relative to Refs.~\onlinecite{ahn2026riemannian,park2026sqls}, our contributions are: (i) a characterization of the configurations reachable by exact time evolution, the support of $e^{-iAt}b$ equals the Krylov support of $b$ except at countably many times, which implies that the Galerkin solution is exact whenever the sampled set contains this support (Proposition~\ref{prop:krylov}); (ii) an analysis of the cost of the method in the Pauli-string encoding used by SKQD, showing that banded matrices require $N$ Pauli strings in the binary encoding and that Lie-Trotter errors populate configurations outside the Krylov support, so that the compression is lost before the accuracy; (iii) an open-source Qiskit implementation with a single interface to noiseless and noisy simulators and to IBM Quantum hardware; and (iv) a comparison with classical Krylov, residual-greedy, and oracle selection on a localized model problem, using exact sampling distributions. We find that the method selects near-optimal subspaces and is robust to readout noise but, consistent with Ref.~\onlinecite{ahn2026riemannian}, that its natural regime (exponentially large $N$, a matrix given as a sum of polynomially many Pauli strings, and a solution concentrated on few configurations) is also a regime in which classical residual-based selection is efficient, so that a quantum advantage requires additional structure. Identifying such structure on available hardware and on more structured problem instances is the subject of ongoing work.

\section{Theoretical procedure}
\label{sec:theory}

The workflow of the procedure is summarized in Fig.~\ref{fig:workflow}; the individual steps are derived below.

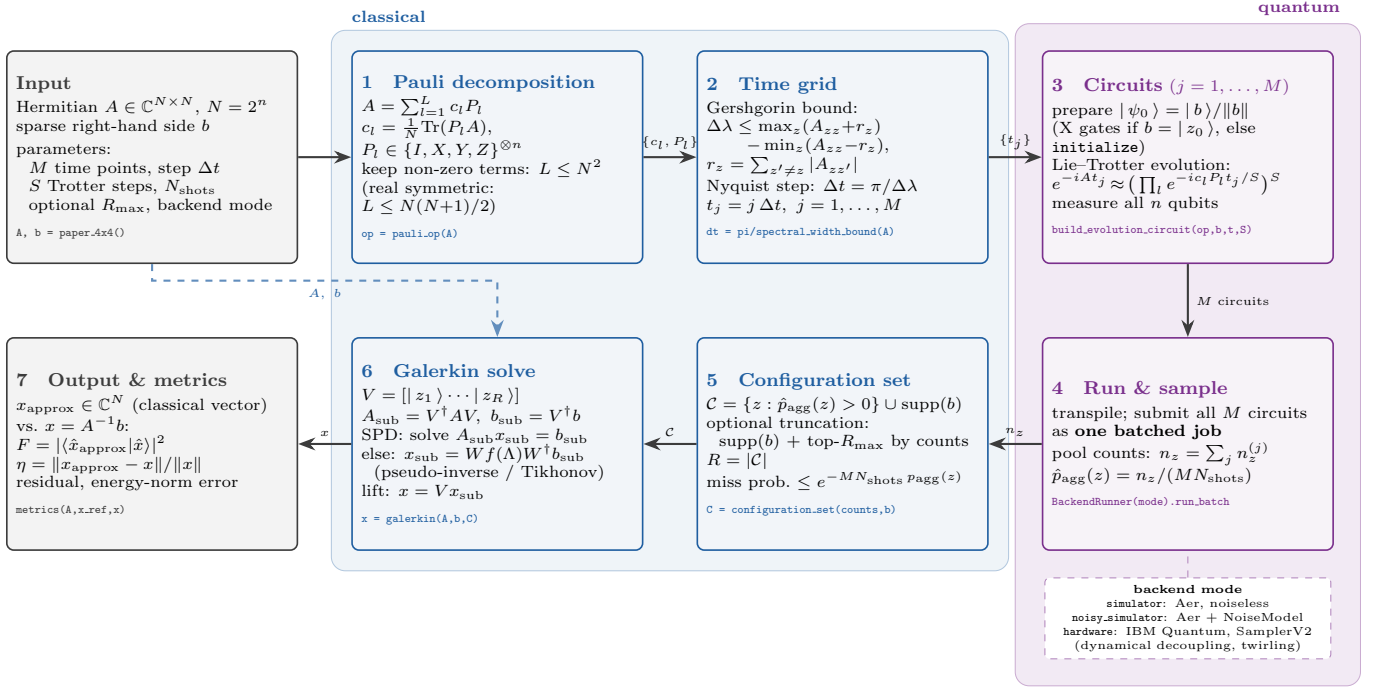
\begin{figure*}[t]
\centering
\resizebox{\textwidth}{!}{%
\begin{tikzpicture}[
  font=\scriptsize, node distance=7mm and 8mm,
  step/.style={draw=#1, line width=0.7pt, rounded corners=3pt, fill=#1!6,
               text width=41mm, align=flush left, inner sep=4pt, minimum height=32mm},
  head/.style={font=\bfseries\footnotesize},
  fn/.style={font=\ttfamily\tiny, text=#1!80!black},
  arr/.style={-{Stealth[length=2.6mm]}, line width=0.9pt, draw=sqdio},
  dec/.style={draw=#1, fill=white, diamond, aspect=2.2, inner sep=1pt, font=\tiny, align=center}
]
\node[step=sqdio] (in) {%
  {\color{sqdio}\textbf{\footnotesize Input}}\\[2pt]
  Hermitian $A\in\mathbb{C}^{N\times N}$, $N=2^n$\\
  sparse right-hand side $b$\\[2pt]
  parameters:\\ \ \ $M$ time points, step $\Delta t$\\ \ \ $S$ Trotter steps, $N_{\rm shots}$\\ \ \ optional $R_{\max}$, backend mode\\[3pt]
  {\ttfamily\tiny\color{sqdio} A, b = paper\_4x4()}};
\node[step=sqdcl, right=of in] (s1) {%
  {\color{sqdcl}\textbf{\footnotesize 1\;\; Pauli decomposition}}\\[2pt]
  $A=\sum_{l=1}^{L}c_lP_l$\\
  $c_l=\tfrac{1}{N}\mathrm{Tr}(P_lA)$, \ $P_l\in\{I,X,Y,Z\}^{\otimes n}$\\
  keep non-zero terms: $L\le N^2$\\ (real symmetric: $L\le N(N{+}1)/2$)\\[3pt]
  {\ttfamily\tiny\color{sqdcl} op = pauli\_op(A)}};
\node[step=sqdcl, right=of s1] (s2) {%
  {\color{sqdcl}\textbf{\footnotesize 2\;\; Time grid}}\\[2pt]
  Gershgorin bound:\\
  $\Delta\lambda\le\max_z(A_{zz}{+}r_z)$\\ \hspace*{6mm}$-\min_z(A_{zz}{-}r_z)$, \ $r_z=\sum_{z'\neq z}|A_{zz'}|$\\
  Nyquist step: $\Delta t=\pi/\Delta\lambda$\\
  $t_j=j\,\Delta t,\ j=1,\dots,M$\\[3pt]
  {\ttfamily\tiny\color{sqdcl} dt = pi/spectral\_width\_bound(A)}};
\node[step=sqdqu, right=of s2] (s3) {%
  {\color{sqdqu}\textbf{\footnotesize 3\;\; Circuits} \normalfont($j=1,\dots,M$)}\\[2pt]
  prepare $\ket{\psi_0}=\ket{b}/\lVert b\rVert$\\
  (X gates if $b=\ket{z_0}$, else \texttt{initialize})\\
  Lie--Trotter evolution:\\
  $e^{-iAt_j}\!\approx\!\big(\prod_l e^{-ic_lP_lt_j/S}\big)^{S}$\\
  measure all $n$ qubits\\[3pt]
  {\ttfamily\tiny\color{sqdqu} build\_evolution\_circuit(op,b,t,S)}};
\node[step=sqdqu, below=11mm of s3] (s4) {%
  {\color{sqdqu}\textbf{\footnotesize 4\;\; Run \& sample}}\\[2pt]
  transpile; submit all $M$ circuits\\ as \textbf{one batched job}\\
  pool counts: $n_z=\sum_j n_z^{(j)}$\\
  $\hat p_{\rm agg}(z)=n_z/(MN_{\rm shots})$\\[3pt]
  {\ttfamily\tiny\color{sqdqu} BackendRunner(mode).run\_batch}};
\node[step=sqdcl, left=of s4] (s5) {%
  {\color{sqdcl}\textbf{\footnotesize 5\;\; Configuration set}}\\[2pt]
  $\mathcal{C}=\{z:\hat p_{\rm agg}(z)>0\}\cup\mathrm{supp}(b)$\\
  optional truncation:\\ \ \ $\mathrm{supp}(b)$ + top-$R_{\max}$ by counts\\
  $R=|\mathcal{C}|$\\
  miss prob.\ $\le e^{-MN_{\rm shots}\,p_{\rm agg}(z)}$\\[3pt]
  {\ttfamily\tiny\color{sqdcl} C = configuration\_set(counts,b)}};
\node[step=sqdcl, left=of s5] (s6) {%
  {\color{sqdcl}\textbf{\footnotesize 6\;\; Galerkin solve}}\\[2pt]
  $V=[\ket{z_1}\cdots\ket{z_R}]$\\
  $A_{\rm sub}=V^{\dagger}AV,\ b_{\rm sub}=V^{\dagger}b$\\
  SPD: solve $A_{\rm sub}x_{\rm sub}=b_{\rm sub}$\\
  else: $x_{\rm sub}=Wf(\Lambda)W^{\dagger}b_{\rm sub}$\\ \ \ (pseudo-inverse / Tikhonov)\\
  lift: $x=Vx_{\rm sub}$\\[3pt]
  {\ttfamily\tiny\color{sqdcl} x = galerkin(A,b,C)}};
\node[step=sqdio, left=of s6] (out) {%
  {\color{sqdio}\textbf{\footnotesize 7\;\; Output \& metrics}}\\[2pt]
  $x_{\rm approx}\in\mathbb{C}^N$ (classical vector)\\
  vs.\ $x=A^{-1}b$:\\
  $F=|\langle\hat x_{\rm approx}|\hat x\rangle|^2$\\
  $\eta=\lVert x_{\rm approx}-x\rVert/\lVert x\rVert$\\
  residual, energy-norm error\\[3pt]
  {\ttfamily\tiny\color{sqdio} metrics(A,x\_ref,x)}};
\draw[arr] (in) -- (s1);
\draw[arr] (s1) -- node[above,font=\tiny]{$\{c_l,P_l\}$} (s2);
\draw[arr] (s2) -- node[above,font=\tiny]{$\{t_j\}$} (s3);
\draw[arr] (s3) -- node[right,font=\tiny,align=left]{$M$ circuits} (s4);
\draw[arr] (s4) -- node[above,font=\tiny]{$n_z$} (s5);
\draw[arr] (s5) -- node[above,font=\tiny]{$\mathcal C$} (s6);
\draw[arr] (s6) -- node[above,font=\tiny]{$x$} (out);
\draw[arr, dashed, draw=sqdcl!70] (in.south) -- ++(0,-2.5mm) -| node[pos=0.25,below,font=\tiny,text=sqdcl]{$A,\ b$} (s6.north);
\node[below=4mm of s4, xshift=0mm, align=center, font=\tiny, text width=41mm, draw=sqdqu!60, dashed, rounded corners=2pt, fill=white, inner sep=3pt] (be) {%
  \textbf{backend mode}\\ \texttt{simulator}: Aer, noiseless\\ \texttt{noisy\_simulator}: Aer + NoiseModel\\ \texttt{hardware}: IBM Quantum, SamplerV2\\ (dynamical decoupling, twirling)};
\draw[sqdqu!60, dashed] (be.north) -- (s4.south);
\begin{scope}[on background layer]
  \node[fit=(s3)(s4)(be), inner sep=4mm, rounded corners=6pt, fill=sqdqu!10, draw=sqdqu!40] (qlane) {};
  \node[anchor=south east, font=\bfseries\scriptsize, text=sqdqu] at ([xshift=-2mm]qlane.north east) {quantum};
  \node[fit=(s1)(s2)(s5)(s6), inner sep=3mm, rounded corners=6pt, fill=sqdcl!7, draw=sqdcl!35] (clane) {};
  \node[anchor=south west, font=\bfseries\scriptsize, text=sqdcl] at ([xshift=2mm]clane.north west) {classical};
\end{scope}

\end{tikzpicture}}
\caption{Workflow of sample-based quantum diagonalization for linear systems. Steps~1--2 (classical) decompose $A$ into Pauli strings and choose the time grid; Steps~3--4 (quantum) prepare $\ket{b}$, apply the Trotterized evolution $e^{-iAt_j}$ for $j=1,\dots,M$, and sample all circuits in one batched job on a noiseless simulator, a noisy simulator, or IBM Quantum hardware; Steps~5--7 (classical) form the configuration set, solve the projected system, and evaluate the solution. Names in typewriter font refer to the functions of the accompanying code.}
\label{fig:workflow}
\end{figure*}

\subsection{Pauli decomposition}
For $N=2^n$, the tensor products $P\in\{I,X,Y,Z\}^{\otimes n}$ form an orthogonal basis for Hermitian matrices $A\in\mathbb{C}^{N\times N}$, with $\Tr(P^\dagger P')=N\delta_{PP'}$:
\begin{equation}
A=\sum_{l=1}^{L}c_lP_l,\qquad c_l=\frac{1}{N}\Tr(P_l^\dagger A)\in\mathbb{R}.
\label{eq:pauli}
\end{equation}
The number $L$ of non-zero terms controls the circuit cost. 
For a real symmetric $A$, strings with an odd number of $Y$ factors are purely imaginary and do not contribute, so a generic real symmetric matrix has $L=(4^n+2^n)/2=N(N+1)/2$ terms and a generic complex Hermitian matrix has $L=N^2$, whereas a diagonal matrix has only $L=N$ terms. These are the maximum numbers of non-zero terms for the respective classes. As in Hamiltonian simulation, truncating the expansion reduces the circuit cost at the expense of an approximation error, and this trade-off has to be tested for the problem at hand.
Equation~\eqref{eq:pauli} can be evaluated in $O(N^2\log N)$ operations with the tensorized Pauli decomposition~\cite{hantzko2024tpd}. We typically use symmetric positive-definite (SPD) matrices in our numerical experiments. If $A$ is not SPD, the decomposition still applies, but the classical solve requires additional care for negative (or zero) eigenvalues, as discussed in the Galerkin step below.

\subsection{State preparation and Trotterized evolution}
This step involves the encoding of the vector $\ket{b}$ into a quantum state and the subsequent Trotterized evolution under the Hamiltonian $A$.
The normalized state $\ket{\psi_0}=\ket{b}/\lVert b\rVert$ is expanded in the eigenbasis $A\ket{\phi_k}=\lambda_k\ket{\phi_k}$,
\begin{equation}
\ket{\psi_0}=\sum_k\beta_k\ket{\phi_k},
\label{eq:eigenbasis-expansion}
\end{equation}
with overlap amplitudes
\begin{equation}
\beta_k=\braket{\phi_k}{\psi_0},\qquad \sum_k|\beta_k|^2=1,
\label{eq:beta-def}
\end{equation}
and evolves as
\begin{equation}
\ket{\psi(t)}=e^{-iAt}\ket{\psi_0}=\sum_k\beta_k e^{-i\lambda_kt}\ket{\phi_k}.
\label{eq:evolved-state}
\end{equation}
The propagator is approximated by the first-order Lie--Trotter product formula~\cite{trotter1959product,lloyd1996universal},
\begin{equation}
e^{-iAt}\approx\Big(\prod_{l=1}^{L}e^{-ic_lP_lt/S}\Big)^{S},
\label{eq:trotter}
\end{equation}
whose error is bounded by the commutators of the terms~\cite{childs2021trotter},
\begin{equation}
\epsilon_{\mathrm{Trot}}\le\frac{t^2}{2S}\sum_{j<l}|c_jc_l|\,\big\lVert[P_j,P_l]\big\rVert,
\label{eq:trotter-bound}
\end{equation}
hence the formula is exact when all terms commute. Since $P^2=\mathbb{I}$, each factor is exactly a Pauli rotation,
\begin{equation}
e^{-i\theta P}=\cos\theta\,\mathbb{I}-i\sin\theta\,P,
\label{eq:pauli-rotation}
\end{equation}
which for a string of weight $w$ costs $2(w-1)$ CNOT gates; the identity string contributes a global phase only. A circuit therefore contains at most $LS$ rotations and $2(n-1)LS$ CNOT gates.

\subsection{Sampling and the configuration set}
\label{sec:sampling}
Measuring $\ket{\psi(t_j)}$ in the computational basis yields configuration $z\in\{0,1\}^n$ with probability $p(z\mid t_j)=|\braket{z}{\psi(t_j)}|^2$. 
Since $p(z\mid t)$ depends on the eigenvalues only through their differences, it is a trigonometric polynomial in $t$ whose angular frequencies lie in $[-\Delta\lambda,\Delta\lambda]$, with $\Delta\lambda=\lambda_{\max}-\lambda_{\min}$. Sampling it on a uniform grid of step $\Delta t=t_{\max}/M$ without aliasing requires
\begin{equation}
\Delta t\lesssim\frac{\pi}{\Delta\lambda}
\quad\Longleftrightarrow\quad M\gtrsim\frac{t_{\max}\Delta\lambda}{\pi}.
\label{eq:nyquist}
\end{equation}
This is a heuristic criterion rather than a guarantee: it excludes aliasing of the sampled distribution, but it does not by itself ensure that every relevant configuration is sampled. Rigorous support-discovery guarantees for finite nonnegative- and signed-time grids are derived in Ref.~\onlinecite{park2026sqls}. For sparse $A$, an inexpensive upper bound on $\Delta\lambda$ follows from the Gershgorin theorem,
\begin{equation}
\Delta\lambda\le\max_z(A_{zz}+r_z)-\min_z(A_{zz}-r_z),
\label{eq:gershgorin}
\end{equation}
with $r_z=\sum_{z'\neq z}|A_{z'z}|$. 
$A_{zz}$ is the expectation value of $A$ in the computational basis state $\ket{z}$:
\begin{equation}
A_{zz} = \bra{z}A\ket{z}.
\end{equation}
In this work, we label the computational basis states by the bit strings $z\in\{0,1\}^n$.
The usual integer index is $z$ read as a binary number, and $A_{z'z}=\langle z'|A|z\rangle$ is the entry in row $z'$ and column $z$, which is off-diagonal when $z'\neq z$. Evaluating Eq.~\eqref{eq:gershgorin} requires the rows of $A$; when $A$ is given as a Pauli sum, the bound $\Delta\lambda\le2\sum_{P_l\neq\mathbb{I}}|c_l|$ can be evaluated in $O(L)$ operations.

The SKQD choice $\Delta t\sim\pi/\lVert A\rVert$,
where $\lVert A\rVert=\max(|\lambda_{\max}|,|\lambda_{\min}|)$, the operator (spectral) norm, i.e. the largest absolute eigenvalue, is related but not equivalent to Eq.~\eqref{eq:nyquist}~\cite{yu2025skqd}: for positive semidefinite $A$ one has $\Delta\lambda\le\lVert A\rVert$, so that it satisfies Eq.~\eqref{eq:nyquist} automatically (and is much finer than necessary when $\Delta\lambda\ll\lVert A\rVert$), whereas in general $\Delta\lambda\le2\lVert A\rVert$, so that it can exceed the bound of Eq.~\eqref{eq:nyquist} by at most a factor of two. Amplitude propagates through the sparsity graph of $A$, so the configurations sampled up to $t_{\max}$ form a light cone around $\supp(b)$ whose radius, and hence the subspace size, is controlled by $t_{\max}$ as we expect to see in the computational examples.

Components of $p(z\mid t)$ whose frequencies differ by less than $\sim2\pi/t_{\max}$ are, moreover, not resolved within the time window, and each $p(z\mid t_j)$ is only estimated up to shot noise $\sim\sqrt{p(1-p)/N_{\mathrm{shots}}}$. As in SQD and SKQD~\cite{robledomoreno2025chemistry,yu2025skqd}, the completeness of the configuration set is therefore not guaranteed a priori; it can be monitored a posteriori through the residual $\lVert b-Ax_{\mathrm{approx}}\rVert$ and improved by increasing $t_{\max}$, $M$, or $N_{\mathrm{shots}}$. This is tested numerically in Sec.~\ref{sec:results}.

The exact aggregated distribution can be defined as the average over the $M$ sampling times,
\begin{equation}
p_{\mathrm{agg}}(z)=\frac{1}{M}\sum_{j=1}^{M}p(z\mid t_j),
\label{eq:pagg}
\end{equation}
which differs from the long-time average $\bar p$ of Remark~\ref{rem:sampling}. Pooling the $N_{\mathrm{shots}}$ samples of each of the $M$ circuits gives its empirical estimate $\hat p_{\mathrm{agg}}(z)$ and the configuration set
\begin{equation}
\mathcal{C}=\{z\mid\hat p_{\mathrm{agg}}(z)>0\}\cup\supp(b),\qquad R=|\mathcal{C}|.
\label{eq:config-set}
\end{equation}
The support of $b$, which is known classically, is always included: a set that misses it gives $b_{\mathrm{sub}}=0$ below and hence a vanishing solution. 
More explicitly, under amplitude encoding, $\supp(b)$ is the set of configurations occupied by $\ket{\psi_0}$ at $t=0$, the initial state. That is, the time evolution spreads amplitude from these configurations through the couplings of $A$.

A configuration $z\notin\supp(b)$ is missed only if it is absent from all $MN_{\mathrm{shots}}$ independent samples, so that, using $1-p\le e^{-p}$,
\begin{equation}
\begin{split}
\Pr[z\notin\mathcal{C}]&=\prod_{j=1}^{M}\big(1-p(z\mid t_j)\big)^{N_{\mathrm{shots}}}\\
&\le\exp\!\big(-MN_{\mathrm{shots}}\,p_{\mathrm{agg}}(z)\big).
\end{split}
\label{eq:capture}
\end{equation}
Hence $R<N$ arises only from configurations of weight below $\sim1/(MN_{\mathrm{shots}})$ or from explicit truncation. The Trotterized propagator does not preserve the sparsity pattern of $A$, since individual Pauli strings can couple configurations that $A$ does not couple; by Eq.~\eqref{eq:capture}, with $p(z\mid t_j)$ now the distributions of the Trotterized circuits, Trotter errors therefore populate spurious configurations unless their weight is below $\sim1/(MN_{\mathrm{shots}})$. The sampling step is summarized in Algorithm~\ref{alg:trotter-sampling}.

\medskip
\refstepcounter{algorithm}
\noindent\textbf{Algorithm \thealgorithm.} Trotterized time-evolution sampling\label{alg:trotter-sampling}
\begin{algorithmic}[1]
\Require Hermitian $A\in\mathbb{C}^{N\times N}$ with Pauli decomposition
  $A=\sum_{l=1}^{L}c_lP_l$; initial state $\ket{\psi_0}$ corresponding to $b$; number of time
  samples $M$; Trotter steps $S$; maximum time $t_{\max}$; shots per
  circuit $N_{\text{shots}}$
\Ensure Empirical aggregated distribution $\hat p_{\mathrm{agg}}(z)$ over
  $z\in\{0,1\}^n$, $n=\log_2 N$
 
\State $\hat p_{\mathrm{agg}}(z) \gets 0$ for all $z$
\For{$j = 1$ to $M$}
    \State $t_j \gets j\cdot t_{\max}/M$ \Comment{time grid}
    \State normalize the state $\ket{\psi_0}$
    \State initialize circuit $\mathcal{Q}_j$ on $n$ qubits in state $\ket{\psi_0}$
    \For{$s = 1$ to $S$} \Comment{Trotter steps}
        \For{$l = 1$ to $L$}
            \State $\theta \gets c_l\,t_j/S$
            \State append $e^{-i\theta P_l}=\cos\theta\,\mathbb{I}-i\sin\theta\,P_l$ to $\mathcal{Q}_j$
        \EndFor
    \EndFor
    \State measure all qubits of $\mathcal{Q}_j$ in the computational basis, $N_{\text{shots}}$ times
    \State record counts $\{n_z^{(j)}\}_z$; set $\hat p(z\mid t_j)\gets n_z^{(j)}/N_{\text{shots}}$
    \State $\hat p_{\mathrm{agg}}(z) \gets \hat p_{\mathrm{agg}}(z) + \hat p(z\mid t_j)/M$ for all $z$
\EndFor
\State \Return{$\hat p_{\mathrm{agg}}(z)$}
\end{algorithmic}
 
Total gate count follows directly from the nested loops: $L$ rotations
per Trotter step, $S$ steps per circuit, $M$ independent circuits, for
$M\times L\times S$ rotation gates overall --- all $M$ circuits are
independent and are submitted as a single batched job in the
implementation below, so $M$ affects total shot count and job size but
not per-circuit depth. The implementation is described in 
Sec.~\ref{sec:implementation}.

\subsection{Galerkin solution}
Let $\mathcal{C}=\{z_1,\ldots,z_R\}$ and let
\begin{equation}
V=\big[\ket{z_1}\cdots\ket{z_R}\big]\in\mathbb{C}^{N\times R}
\label{eq:projector-def}
\end{equation}
be the corresponding isometry ($V^\dagger V=\mathbb{I}_R$; the orthogonal projector onto the subspace is $VV^\dagger$). The system is restricted to the subspace,
\begin{equation}
A_{\mathrm{sub}}=V^\dagger AV,\qquad b_{\mathrm{sub}}=V^\dagger b,
\label{eq:subspace-restriction}
\end{equation}
which in practice amounts to selecting rows and columns of $A$. For SPD $A$, $A_{\mathrm{sub}}$ is SPD and the reduced system is solved directly. The spectral decomposition $A_{\mathrm{sub}}=W\Lambda W^\dagger$ gives
\begin{equation}
x_{\mathrm{sub}}=W\,f(\Lambda)\,W^\dagger b_{\mathrm{sub}},
\label{eq:spectral-inversion}
\end{equation}
with $f(\lambda)=1/\lambda$, replaced for indefinite $A$ by a pseudo-inverse or the Tikhonov filter $f(\lambda)=\lambda/(\lambda^2+\tau^2)$. This is the diagonalization step of the method; it carries out classically the $1/\lambda$ inversion that HHL performs coherently. The approximate solution is $x_{\mathrm{approx}}=Vx_{\mathrm{sub}}$, and its accuracy is quantified by the fidelity and the relative error,
\begin{equation}
F=\big|\braket{\hat x_{\mathrm{approx}}}{\hat x}\big|^2,\qquad
\eta=\frac{\lVert x_{\mathrm{approx}}-x\rVert}{\lVert x\rVert},
\label{eq:metrics}
\end{equation}
and by the energy-norm error $\lVert x_{\mathrm{approx}}-x\rVert_A$,
where $x\equiv x_{\mathrm{classical}}=A^{-1}b$, hats denote normalization, and $\lVert v\rVert_A^2=v^\dagger Av$. Unlike a comparison of $|x_z|^2$ alone, these metrics are sensitive to the signs of the components.

\subsection{Exactness, optimality, and targeting}
\label{sec:propositions}
We denote by $\supp(v)$ the set of configurations on which $v$ is non-zero and assume $A$ invertible.

\begin{proposition}[Krylov support]\hfill\\
\label{prop:krylov}
Let $\mathcal{S}_{\mathrm{Kr}} = \bigcup_{m=0}^{N-1}\supp(A^m b)$ be the Krylov support of
$b$. Then (i) $\supp(x_{\mathrm{classical}}) \subseteq \mathcal{S}_{\mathrm{Kr}}$;
(ii) $\supp(e^{-iAt}b) = \mathcal{S}_{\mathrm{Kr}}$ for all $t$ outside a countable set;
(iii) if $\mathcal{C}\supseteq\mathcal{S}_{\mathrm{Kr}}$ and $A_{\mathrm{sub}}$ is
invertible, then $x_{\mathrm{approx}} = x_{\mathrm{classical}}$ exactly.
\end{proposition}

\textit{Proof.} (i) By the Cayley--Hamilton theorem, $A^{-1}$ is a polynomial of degree at most $N-1$ in $A$, so $x_{\mathrm{classical}}=A^{-1}b\in\mathrm{span}\{A^mb\}_{m=0}^{N-1}$. (ii) Each amplitude $f_z(t)=\braket{z}{e^{-iAt}b}=\sum_{m\ge0}\frac{(-it)^m}{m!}\braket{z}{A^m b}$ is an entire function of $t$; it vanishes identically for $z\notin\mathcal{S}_{\mathrm{Kr}}$ (coefficients with $m\ge N$ are combinations of those with $m<N$), and otherwise has isolated zeros only. (iii) By (i), $x=VV^\dagger x$ and $b=VV^\dagger b$, so $A_{\mathrm{sub}}V^\dagger x=b_{\mathrm{sub}}$ and $x_{\mathrm{approx}}=VV^\dagger x=x$. $\square$

By (ii), the states $e^{-iAt_j}\ket{b}$ explore exactly the configurations that can carry weight in the solution. If the graph of $A$ is connected, $\mathcal{S}_{\mathrm{Kr}}$ generically contains all $N$ configurations, so $R<N$ arises only from finite sampling [Eq.~\eqref{eq:capture}] or truncation, and the accuracy depends on how fast the solution decays away from the sampled configurations. A Trotterized propagator need not preserve $\mathcal{S}_{\mathrm{Kr}}$ (Sec.~\ref{sec:noise-trotter}).

\begin{remark}[Galerkin optimality]\hfill\\
\label{rem:galerkin}
For SPD $A$, $x_{\mathrm{approx}}=VA_{\mathrm{sub}}^{-1}V^\dagger b$ minimizes $\lVert x_{\mathrm{classical}}-y\rVert_A$ over all $y$ supported in $\mathcal{C}$, so the error cannot increase when $\mathcal{C}$ is enlarged; this is the standard Galerkin property, and quantitative bounds under random sampling are given in Ref.~\onlinecite{ahn2026riemannian}. Noise-induced configurations therefore cannot increase the energy-norm error as long as $\mathcal{C}_{\mathrm{noisy}}\supseteq\mathcal{C}_{\mathrm{ideal}}$; noise degrades the accuracy only when this nesting fails or when $\mathcal{C}$ is truncated.
\end{remark}

\begin{remark}[Sampling distribution]\hfill\\
\label{rem:sampling}
For a non-degenerate spectrum, the long-time average of $p(z\mid t)$ is $\bar p(z)=\sum_k|\beta_k|^2|\braket{z}{\phi_k}|^2$, whereas $x_{\mathrm{classical},z}=\lVert b\rVert\sum_k\beta_k\lambda_k^{-1}\braket{z}{\phi_k}$. Sampling weights each eigencomponent by $|\beta_k|^2$, the solution by $\beta_k/\lambda_k$, so the most frequently sampled configurations need not carry the solution, increasingly so for large condition numbers~\cite{park2026sqls}; on a finite grid, cross terms with $|\lambda_k-\lambda_{k'}|\lesssim\pi/t_{\max}$ are not averaged out. Residual-based selection is a natural alternative (Sec.~\ref{sec:classical}).
\end{remark}

\subsection{Relation to HHL and to SQD}
\label{sec:hhl-comparison}
Compared with HHL, the method requires no phase estimation, no ancilla, and no post-selection, and it returns the full classical vector $x_{\mathrm{approx}}$; the price is a classical solve of dimension $R$. Both methods require the preparation of $\ket{b}$, which costs $O(N)$ gates unless $b$ has exploitable structure. Compared with SQD for chemistry, the method shares the sample-then-project structure and uses the SKQD sampling protocol, but the classical step is a Galerkin linear solve, in which diagonalization is one option rather than the goal, and no configuration recovery is needed~\cite{robledomoreno2025chemistry,yu2025skqd}. Compared with the sample-based quantum linear solvers of Refs.~\onlinecite{ahn2026riemannian,park2026sqls}, the method shares the measurement of time-evolved (Krylov) states and, with Ref.~\onlinecite{ahn2026riemannian}, the Galerkin reconstruction; it differs in implementing the evolution by Lie--Trotter products of the Pauli decomposition, whose circuit cost and leakage we analyze, whereas Ref.~\onlinecite{park2026sqls} analyzes qDRIFT compilation.

\section{Implementation}
\label{sec:implementation}
The method is implemented in the Python package \texttt{sqd\_linsolve}, built on Qiskit~\cite{qiskit} and Qiskit Aer (tested with Qiskit 2.3.0 and 2.5.2 and Aer 0.17.2, with identical results), and in a self-contained notebook that produces the $4\times4$ example, the hardware run, and the scalability test of Figs.~\ref{fig:hardware_vs_simulation} and~\ref{fig:scalability}. The package, the notebook, and scripts that reproduce the simulator tests of Secs.~\ref{sec:localized}, \ref{sec:noise-trotter}, \ref{sec:dense}, and~\ref{sec:classical} are provided as ancillary files with this article; the hardware run requires an IBM Quantum account. The $M$ circuits of Algorithm~\ref{alg:trotter-sampling} are built with \texttt{PauliEvolutionGate} and submitted as a single batched job through one interface to a noiseless simulator, a noisy simulator, or IBM Quantum hardware. Gate counts refer to transpilation to $\{\mathrm{CX},R_Z,\sqrt{X},X\}$. To isolate the subspace selection from Trotter errors, the tests of Secs.~\ref{sec:localized}, \ref{sec:dense}, and~\ref{sec:classical} sample the exact distributions $p(z\mid t_j)$; Trotter errors are quantified separately with actual circuits.

\section{Results}
\label{sec:results}

We first validate the pipeline on a $4\times4$ SPD example, using a noiseless simulator, a noisy simulator, and IBM Quantum hardware (Sec.~\ref{sec:validation}). We then test the method on a localized problem with $N=1024$ (Sec.~\ref{sec:localized}), study the effect of readout noise, shots, and Trotter errors (Sec.~\ref{sec:noise-trotter}), quantify the circuit cost and the behavior for dense matrices (Sec.~\ref{sec:dense}), examine the scaling with $N$ (Sec.~\ref{sec:scalability}), and compare with classical selection strategies (Sec.~\ref{sec:classical}).

\subsection{Validation on simulator and hardware}
\label{sec:validation}
We first consider the $4\times4$ SPD matrix with rows $(5,3,1,-2)$, $(3,5,2,-1)$, $(1,2,5,-3)$, $(-2,-1,-3,5)$, eigenvalues $\{1,3,5,11\}$, and $b=(0,0,0,1)^\top$, i.e.\ $\ket{\psi_0}=\ket{11}$. Its Pauli decomposition is $A=5\,II+XZ+2\,YY+3\,ZX$, and since the three non-identity strings commute, Eq.~\eqref{eq:trotter} is exact for any $S$ (the infidelity is at the round-off level, below $3\times10^{-14}$ for all $S$ in Fig.~\ref{fig:s_analysis}). This gives four Pauli terms, compared with the maximum of $N(N+1)/2=10$ for a real symmetric matrix.
We used $M=8$ times $t_j=0.25j$, which satisfies Eq.~\eqref{eq:nyquist} with $\Delta\lambda=\lambda_{\max}-\lambda_{\min}=10$, and $N_{\mathrm{shots}}=2000$. For $S=4$, the circuits contain $24$ CNOT gates and have depth $76$; for $S=1$, they contain $6$ CNOT gates and have depth $22$ (transpiler optimization level 1). Figure~\ref{fig:s_analysis} shows the CNOT gate count, the circuit depth, and the infidelity as a function of $S$.
\begin{figure}[t]
  \centering
  \includegraphics[width=\columnwidth]{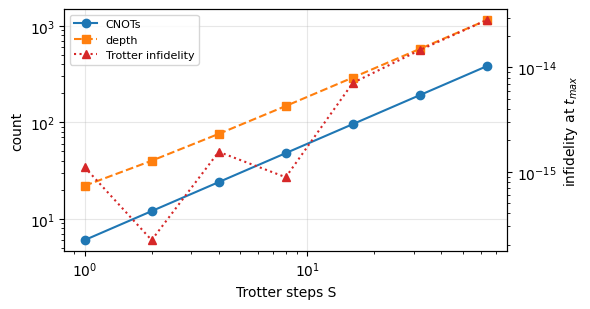}
  \caption{$4\times4$ example: number of CNOT gates and circuit depth (left axis) and Trotter infidelity at $t_{\max}$ (right axis) as a function of the number of Trotter steps $S$. The CNOT count is $6S$. The infidelity stays at the round-off level because the three non-identity Pauli strings commute, so that the Lie--Trotter product is exact for every $S$.}
  \label{fig:s_analysis}
\end{figure}

Since the product formula is exact, both settings produce the same output state. The distribution $p(z\mid t_j)$ varies strongly with time: it concentrates $93.2\%$ on $\ket{00}$ at $t_j=1.5$ and $57.5\%$ on $\ket{10}$ at $t_j=1.0$. The exact distributions at the eight sampling times are listed in Table~\ref{tab:time_distributions}.

\begin{table}[t]
\caption{Time-dependent distributions $p(z\mid t_j)$ for the $4\times4$ example.}
\label{tab:time_distributions}
\begin{ruledtabular}
\begin{tabular}{ccccc}
$t_j$ & $p(\ket{00})$ & $p(\ket{01})$ & $p(\ket{10})$ & $p(\ket{11})$\\
\hline
0.25  &  0.137    & 0.125    & 0.343    & 0.394\\
0.50  &  0.069    & 0.543    & 0.225    & 0.163\\
0.75  &  0.212    & 0.323    & 0.184    & 0.281\\
1.00  &  0.239    & 0.125    & 0.575    & 0.061\\
1.25  &  0.213    & 0.401    & 0.238    & 0.148\\
1.50  &  0.932    & 0.043    & 0.006    & 0.019\\
1.75  &  0.627    & 0.226    & 0.052    & 0.095\\
2.00  &  0.119    & 0.333    & 0.442    & 0.105\\
\end{tabular}
\end{ruledtabular}
\end{table}

The results of the noiseless simulator, where all $M$ circuits run in one batch, are shown in Table~\ref{tab:time_distributions_noiseless}: the counts for each time point, the pooled counts, and the comparison of the empirical $\hat p_{\rm agg}$ with the exact $p_{\rm agg}$. The standard error of the pooled estimate is $\sigma_z=\big\{\sum_j p(z\mid t_j)[1-p(z\mid t_j)]\big\}^{1/2}/(M\sqrt{N_{\mathrm{shots}}})$, because the $M$ circuits are independent binomial experiments.

\begin{table}[t]
\caption{Noiseless \texttt{AerSimulator} counts per time point ($N_{\mathrm{shots}}=2000$), pooled counts, empirical $\hat p_{\mathrm{agg}}$, exact $p_{\mathrm{agg}}$, and pooled standard error $\sigma_z$ for the $4\times4$ example.}
\label{tab:time_distributions_noiseless}
\begin{ruledtabular}
\begin{tabular}{ccccc}
$t_j$ & $\ket{00}$ & $\ket{01}$ & $\ket{10}$ & $\ket{11}$\\
\hline
0.25 & 266  & 259  & 699  & 776\\
0.50 & 127  & 1118 & 422  & 333\\
0.75 & 393  & 663  & 367  & 577\\
1.00 & 484  & 235  & 1159 & 122\\
1.25 & 420  & 781  & 501  & 298\\
1.50 & 1860 & 85   & 11   & 44\\
1.75 & 1266 & 440  & 113  & 181\\
2.00 & 235  & 689  & 850  & 226\\
\hline
Pooled counts & 5051 & 4270 & 4122 & 2557\\
$\hat p_{\mathrm{agg}}$ & 0.3157 & 0.2669 & 0.2576 & 0.1598\\
$p_{\mathrm{agg}}$ & 0.3186 & 0.2650 & 0.2581 & 0.1583\\
$\sigma_z$ & 0.0029 & 0.0033 & 0.0032 & 0.0027\\
\end{tabular}
\end{ruledtabular}
\end{table}

The next step constructs $V=[\ket{z_1}\cdots\ket{z_R}]$, which selects the rows and columns of $A$, solves $A_{\rm sub}x_{\rm sub}=b_{\rm sub}$, and lifts the result back with $x=Vx_{\rm sub}$. All four configurations are observed, so $R=N$, $A_{\rm sub}=A$, and $b_{\rm sub}=b$. The solution $x_{\rm SQD}=(0.1939,\,-0.1394,\,0.2606,\,0.4061)$ coincides with the reference $x_{\rm ref}=A^{-1}b$ to machine precision ($F=1$, $\eta\approx0$).

When $R<N$, as is necessary for large $N$, the ranking of the configurations determines the accuracy. In the following sections we keep $\supp(b)$ plus the most frequently sampled configurations and compare with the oracle, i.e.\ the $R$ configurations with the largest $|x_z|$, which requires knowing the solution.

Figure~\ref{fig:hardware_vs_simulation} compares, for the $4\times4$ example, the
distributions sampled on the noiseless \texttt{AerSimulator}, on
\texttt{AerSimulator} with a noise model, and on the IBM Quantum processor
\texttt{ibm\_fez} with the exact distributions. Panels~(a) and~(c) resolve the
comparison in time: the noiseless simulator follows the exact curves to within
shot noise, whereas the noisy simulator and \texttt{ibm\_fez} deviate by up to
$0.18$ and $0.28$ in total variation distance, respectively, with the largest
deviations at $t_j=1.5$, where the exact distribution is most concentrated.
Panel~(b) shows the corresponding bias of the aggregated distribution
$p_{\mathrm{agg}}(z)$ of Eq.~\eqref{eq:pagg} towards the uniform distribution.
Since all four configurations are nevertheless sampled on every platform
($R=N$), the Galerkin solution coincides with $A^{-1}b$ irrespective of the
noise [panel~(d), Remark~\ref{rem:galerkin}]. This outcome is guaranteed: by Eq.~\eqref{eq:capture}, missing the least probable configuration has probability below $e^{-2500}$ (with the exact $p_{\mathrm{agg}}=0.158$ and $MN_{\mathrm{shots}}=1.6\times10^4$), so the configuration set exhausts the space on each platform, and by Remark~\ref{rem:galerkin} added configurations cannot reduce the accuracy. The run therefore validates the pipeline and probes the sampled distributions, but not the effect of noise on the solution, which requires the truncated regime $R<N$.

The deviations are quantified by the total variation distance between two
distributions over the same configurations,
\begin{equation}
\mathrm{TVD}(p,q)=\frac12\sum_z\big|p(z)-q(z)\big|,
\label{eq:tvd}
\end{equation}
which is the largest difference between the probabilities that $p$ and $q$
assign to any set of configurations; it vanishes for identical distributions
and equals one for distributions with disjoint supports. Even an ideal sampler
gives $\mathrm{TVD}>0$ because of the finite number of shots. For $N_s$
independent samples of $p$, the expected value is approximately
\begin{equation}
\mathrm{TVD}_{\mathrm{shot}}\approx\frac12\sqrt{\frac{2}{\pi}}\sum_z
\sqrt{\frac{p(z)\,[1-p(z)]}{N_s}},
\label{eq:tvd-shot}
\end{equation}
which is $0.013$ for the individual distributions $p(z\mid t_j)$
($N_s=2000$, averaged over $j$). For the pooled estimate $\hat p_{\mathrm{agg}}$,
whose standard error $\sigma_z$ is given in Table~\ref{tab:time_distributions_noiseless}, the
corresponding value is $0.0048$, and the statistic
$\chi^2=\sum_z[\hat p_{\mathrm{agg}}(z)-p_{\mathrm{agg}}(z)]^2/\sigma_z^2$ has the
expected value $N-1=3\pm2.4$. A TVD at this level means that the platform cannot
be distinguished from an ideal sampler with the given number of shots; a TVD
clearly above it signals hardware noise or Trotter error.

Table~\ref{tab:tvd} summarizes the comparison. The noiseless simulator is at
the shot-noise level ($\mathrm{TVD}=0.0033$, $\chi^2=1.6$). The noisy simulator
and \texttt{ibm\_fez} exceed it by factors of $5$ and $11$ in the aggregated
TVD, with $\chi^2=136$ and $406$ and largest deviations of $9.2$ and
$13.5$ standard deviations, both in $\ket{11}$; their per-time TVDs, $0.077$ and
$0.119$ on average, are six and nine times the shot-noise level. These
deviations are irrelevant for the solution when $R=N$, but they matter in the
truncated regime $R<N$, in which the ranking of configurations by their sampled
frequency determines the subspace.

\begin{table}[t]
\caption{Deviation of the sampled distributions from the exact ones for the
$4\times4$ example of Fig.~\ref{fig:hardware_vs_simulation}:
TVD of the aggregated distribution, mean TVD of the per-time distributions,
$\chi^2$ and largest $|z|$-score of $\hat p_{\mathrm{agg}}$ with the pooled
variance $\sigma_z^2$, size $R$ of the configuration set, and fidelity $F$ of the
solution. Last row: values expected from shot noise alone.}
\label{tab:tvd}
\begin{ruledtabular}
\begin{tabular}{lcccccc}
Platform & $\mathrm{TVD}_{\mathrm{agg}}$ & $\overline{\mathrm{TVD}}_{t_j}$ & $\chi^2$ & $\max|z|$ & $R$ & $F$\\
\hline
Noiseless simulator & 0.0033 & 0.012 & 1.6 & 1.0 & 4 & 1\\
Noisy simulator & 0.0251 & 0.077 & 135.6 & 9.2 & 4 & 1\\
\texttt{ibm\_fez} & 0.0542 & 0.119 & 405.8 & 13.5 & 4 & 1\\
\hline
Shot noise & 0.0048 & 0.013 & $3\pm2.4$ & $\lesssim2$ & -- & --\\
\end{tabular}
\end{ruledtabular}
\end{table}

\begin{figure*}[t]
\centering
\includegraphics[width=\textwidth]{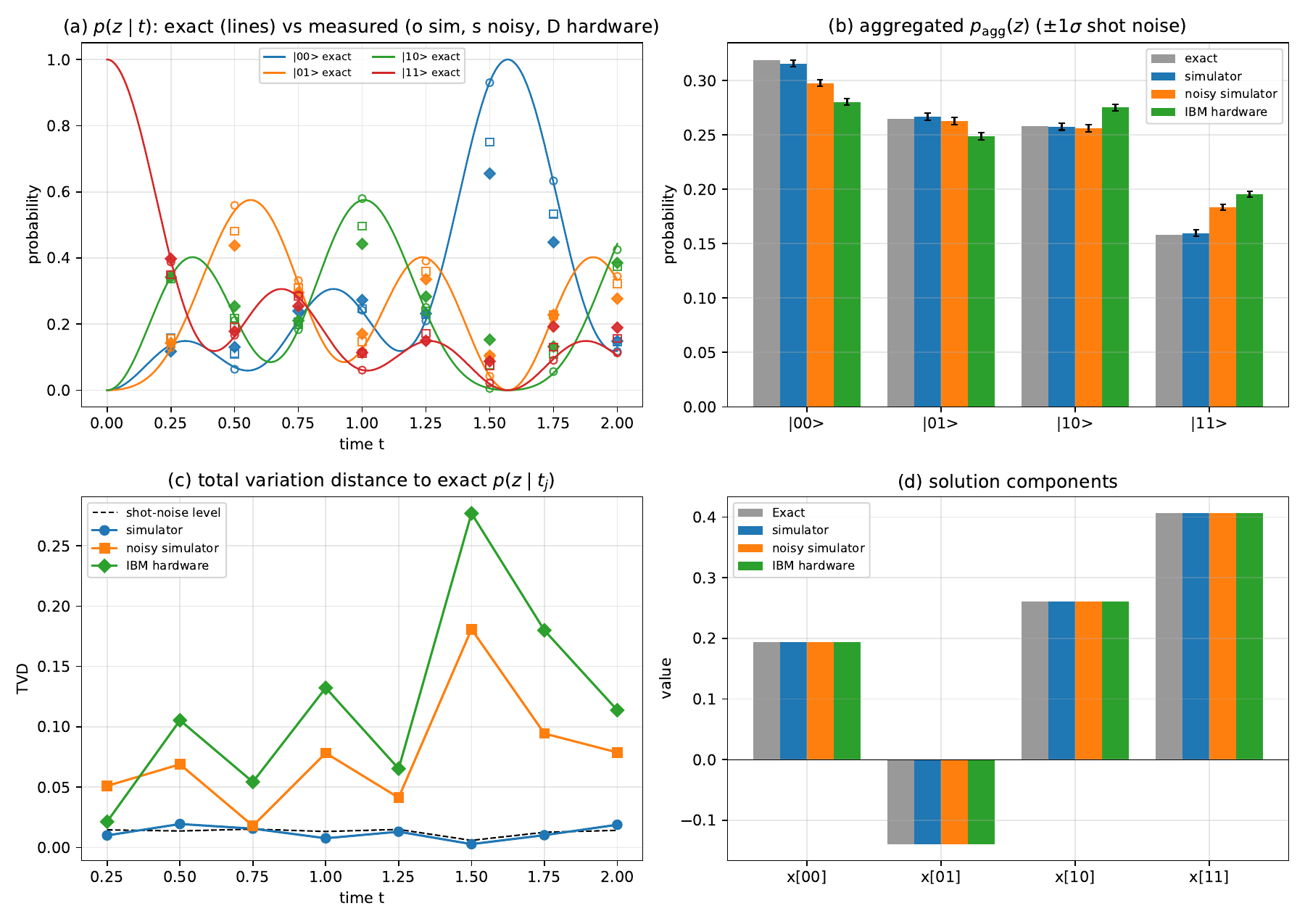}
\caption{Simulation versus hardware for the $4\times4$ example
($M=8$, $S=4$, $N_{\mathrm{shots}}=2000$ per time point): noiseless
\texttt{AerSimulator} (circles), \texttt{AerSimulator} with depolarizing
errors of $10^{-2}$ (two-qubit) and $10^{-3}$ (one-qubit) and $2\%$ readout
errors (squares), and \texttt{ibm\_fez} with dynamical decoupling and Pauli
twirling (diamonds).
(a)~Exact time-dependent probabilities $p(z\mid t)$ (lines) and measured
frequencies at the sampling times $t_j$ (symbols).
(b)~Aggregated distribution $p_{\mathrm{agg}}(z)$, exact and measured, with
$\pm1\sigma$ shot-noise error bars.
(c)~Total variation distance between the measured and the exact
$p(z\mid t_j)$ at each sampling time; the dashed line is the value expected
from shot noise alone.
(d)~Solution components recovered by the Galerkin projection on each
platform, compared with the classical solution $A^{-1}b$.}
\label{fig:hardware_vs_simulation}
\end{figure*}

\subsection{Localized problems}
\label{sec:localized}
To test the regime in which compression is possible, we consider the screened Poisson operator $A=-\Delta_h+\mu\mathbb{I}$, i.e.\ the tridiagonal matrix with $2+\mu$ on the diagonal and $-1$ on the off-diagonals, with $\mu=0.5$ and a point source $b=\ket{z_0}$ at the center, $z_0=N/2$. The solution decays exponentially away from the source (participation ratio $1/\sum_z|\hat x_z|^4=2.45$ at $N=1024$)~\cite{demko1984decay}. Equation~\eqref{eq:gershgorin} gives the bound $\Delta\lambda\le4$ (the true width is $3.99998$ at $n=10$), and we use $\Delta t=\pi/4$, i.e.\ $t_{\max}=2\pi$, $M=8$, and $N_{\mathrm{shots}}=2000$. The results below are those of the accompanying notebook (random seed $0$); the dependence on the seed is stated where it matters.

For $N=1024$, sampling the exact distributions $p(z\mid t_j)$ gives a configuration set of $R=33$ of the $1024$ configurations, the contiguous window $z_0\pm16$, and the solution reaches $F=1-1.2\times10^{-10}$ and $\eta=1.1\times10^{-5}$ (Table~\ref{tab:banded}). Over $30$ seeds, $R$ ranges from $31$ to $35$ (median $33$) and $\eta$ from $0.8\times10^{-5}$ to $2.2\times10^{-5}$. The window is set by $t_{\max}$ and not by $N$: $R=33$ for every $N\ge64$, so that $R/N$ falls from $0.52$ to $0.03$ between $N=64$ and $N=1024$ while $\eta$ stays at $1.1\times10^{-5}$. Quadrupling the time window ($M=32$) enlarges the light cone to $R=109$ for this seed (median $111.5$, range $109$--$114$ over $30$ seeds) and reduces the error to machine precision ($6\times10^{-17}$).

To test whether the compression survives actual circuits, the same problem was run with Lie--Trotter circuits on the noiseless simulator for $N=64$ and $128$, with $S$ the smallest power of two for which the Trotter infidelity at $t_{\max}$ is below $10^{-2}$ ($S=64$, infidelity $2.5\times10^{-3}$ for both sizes). The circuits select $R=32$ configurations, the window of the exact sampling without its outermost configuration and with no configuration outside it, and give $\eta=1.7\times10^{-5}$, at the price of $2.1\times10^4$ and $4.5\times10^4$ CNOT gates per circuit (Table~\ref{tab:banded}). Larger sizes were not run with circuits because of their cost (Sec.~\ref{sec:dense}).

Table~\ref{tab:localized} and Fig.~\ref{fig:tests}(a) compare truncated subspaces for $N=1024$: the most frequently sampled configurations, together with $\supp(b)$, reproduce the oracle choice of the $R$ configurations with the largest $|x_z|$ at every $R$, as does classical greedy selection. Without $\supp(b)$, the two most frequently sampled configurations do not contain the source and give $F=0$; the source alone accounts for $F=0.60$, which is the plateau of random selection with the source included, while random selection without it fails almost completely. Truncating to $R_{\mathrm{trunc}}=N/8$ instead, as in Table~\ref{tab:banded}, selection by sampling coincides with the oracle in every row.

\begin{table*}[t]
\caption{Banded (screened Poisson) family, $\mu=0.5$, point source at $z_0=N/2$, $\Delta t=\pi/\Delta\lambda_{\mathrm{Gersh}}=\pi/4$, $M=8$, $N_{\mathrm{shots}}=2000$, seed $0$ (accompanying notebook). ``Circuits'': noiseless \texttt{AerSimulator}, with $S$ the smallest power of two giving a Trotter infidelity below $10^{-2}$ at $t_{\max}$ (cap $S\le64$); ``exact'': sampling of the exact distributions. $R$ and $R/N$: untruncated configuration set; $\eta$: relative error of the untruncated solution; $F_{\mathrm{trunc}}$: fidelity after truncation to $R_{\mathrm{trunc}}=N/8$ configurations (SQD; the oracle gives the same value to three digits in every row).}
\label{tab:banded}
\begin{ruledtabular}
\begin{tabular}{lcccccccc}
Platform & $N$ & $S$ & CNOTs per circuit & $R$ & $R/N$ & $\eta$ & $R_{\mathrm{trunc}}$ & $F_{\mathrm{trunc}}$\\
\hline
Circuits & 64   & 64 & 21248 & 32 & 0.500 & $1.7\times10^{-5}$ & 8   & 0.995\\
Circuits & 128  & 64 & 45312 & 32 & 0.250 & $1.7\times10^{-5}$ & 16  & 1.000\\
Exact    & 64   & -- & --    & 33 & 0.516 & $1.1\times10^{-5}$ & 8   & 0.995\\
Exact    & 128  & -- & --    & 33 & 0.258 & $1.1\times10^{-5}$ & 16  & 1.000\\
Exact    & 256  & -- & --    & 33 & 0.129 & $1.1\times10^{-5}$ & 32  & 1.000\\
Exact    & 512  & -- & --    & 33 & 0.064 & $1.1\times10^{-5}$ & 64  & 1.000\\
Exact    & 1024 & -- & --    & 33 & 0.032 & $1.1\times10^{-5}$ & 128 & 1.000\\
\end{tabular}
\end{ruledtabular}
\end{table*}

\begin{table*}[t]
\caption{Fidelity of truncated subspaces of size $R$ for the screened Poisson problem ($N=1024$, $\mu=0.5$, $\Delta t=\pi/4$, $M=8$, $N_{\mathrm{shots}}=2000$, exact distributions, seed $0$). SQD: the most frequently sampled configurations with $\supp(b)$ included. Oracle: the $R$ configurations with the largest $|x_z|$. Greedy: classical residual-based selection from $\supp(b)$. Random: mean over $1000$ random sets, with the source included (last column); without the source the mean is $\le0.02$ at all $R$.}
\label{tab:localized}
\begin{ruledtabular}
\begin{tabular}{cccccc}
$R$ & SQD & SQD without $\supp(b)$ & Oracle & Greedy (classical) & Random, with source\\
\hline
2  & 0.745   & 0.000   & 0.745   & 0.745   & 0.600\\
4  & 0.931   & 0.931   & 0.931   & 0.931   & 0.601\\
8  & 0.995   & 0.995   & 0.995   & 0.995   & 0.602\\
16 & 0.99998 & 0.99998 & 0.99998 & 0.99998 & 0.603\\
32 & $>1-10^{-9}$ & $>1-10^{-9}$ & $>1-10^{-9}$ & $>1-10^{-9}$ & 0.611\\
\end{tabular}
\end{ruledtabular}
\end{table*}

\begin{figure*}[t]
\centering
\includegraphics[width=0.92\textwidth]{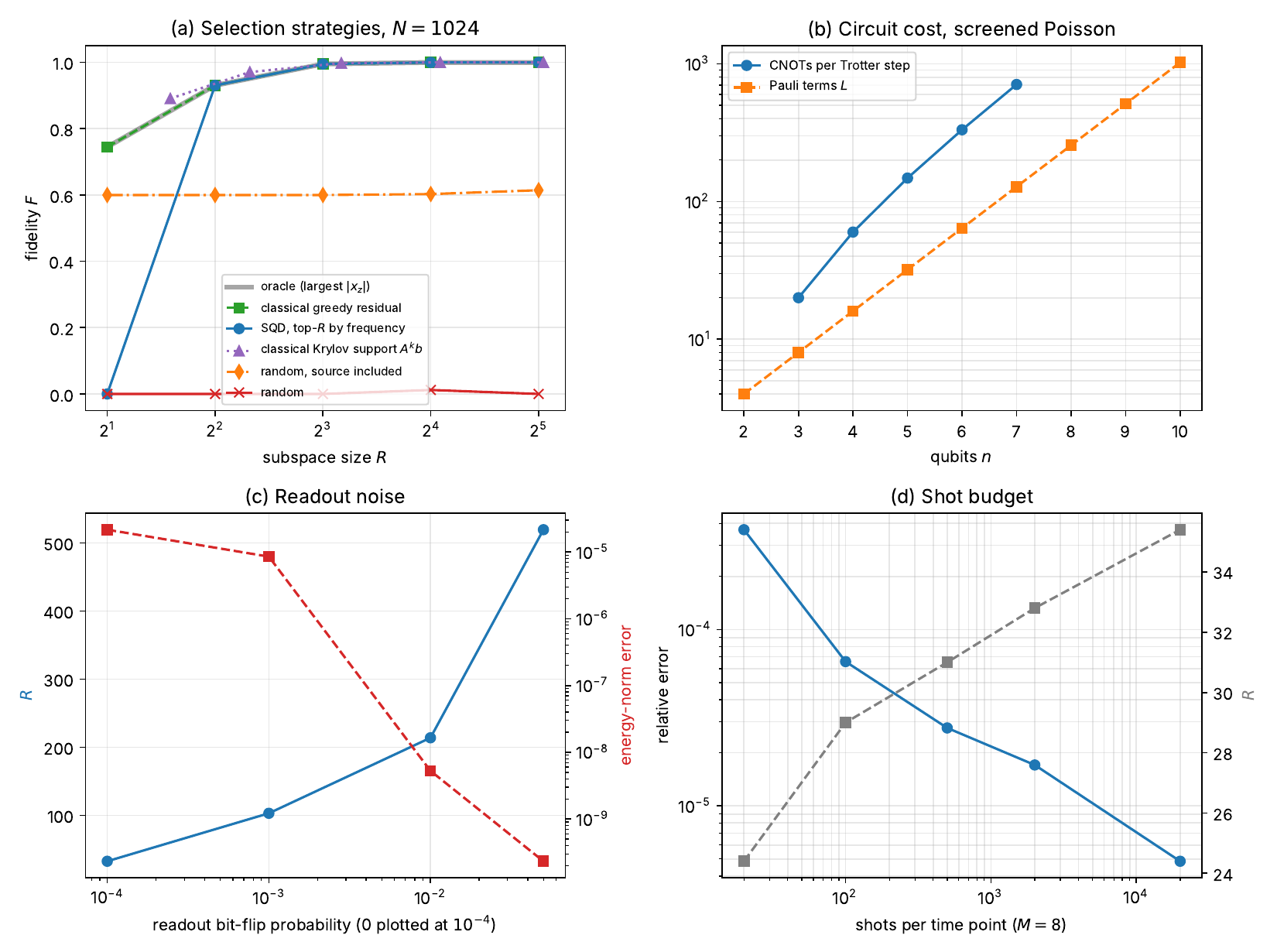}
\caption{(a) Fidelity versus subspace size for the screened Poisson problem of Table~\ref{tab:localized} (SQD shown without the inclusion of $\supp(b)$), compared with the oracle, classical greedy residual selection, the classical Krylov support $\bigcup_{m'\le m}\supp(A^{m'}b)$, and random selection. (b) Pauli terms $L$ and transpiled CNOT gates per Trotter step for the screened Poisson operator. (c) Mean configuration-set size and largest energy-norm error over five repetitions versus the readout bit-flip probability (untruncated). (d) Mean relative error and mean $R$ over five repetitions versus shots per time point (untruncated).}
\label{fig:tests}
\end{figure*}

\subsection{Noise, shots, and Trotter errors}
\label{sec:noise-trotter}
Adding independent readout bit flips with probability $p=10^{-3}$, $10^{-2}$, and $5\times10^{-2}$ enlarges the untruncated configuration set from $R=33$ to $103$, $214$, and $519$, while the energy-norm error decreases from $2\times10^{-5}$ to $2\times10^{-10}$ [Fig.~\ref{fig:tests}(c)], as Remark~\ref{rem:galerkin} predicts; with truncation to $R=16$, the fidelity remains $0.99998$ at all noise levels. The accuracy depends only weakly on the shot budget [Fig.~\ref{fig:tests}(d)]: $20$ shots per time point give $R=24$ and $\eta=4\times10^{-4}$, and $2\times10^4$ give $R=35$ and $5\times10^{-6}$.

To test Proposition~\ref{prop:krylov} with actual circuits, we interleaved two $8\times8$ SPD blocks by a random permutation ($N=16$, $L=136$) and placed $b$ in one block. Exact dynamics select precisely the $8$-configuration Krylov support and give the exact solution with $R=N/2$. Trotterized circuits on the noiseless simulator ($S=2$ and $8$) produce eight spurious configurations outside the support, so $R=N$: the solution remains exact, but the compression is lost. Trotter errors thus degrade the efficiency of the method before its accuracy. This depends on the accuracy of the circuits: for the banded problem of Sec.~\ref{sec:localized}, with $S$ chosen for an infidelity below $10^{-2}$, no configuration outside the light cone appears (Table~\ref{tab:banded}), but at $2\times10^4$--$5\times10^4$ CNOT gates per circuit.

\subsection{Circuit cost and dense matrices}
\label{sec:dense}
For the tridiagonal operator of Sec.~\ref{sec:localized}, the Pauli decomposition contains exactly $L=N$ strings with weights up to $n$, because each hopping term $\ket{z}\bra{z+1}$ acts on the higher bits through projectors that expand into exponentially many strings. The number of CNOT gates per Trotter step, $20$, $60$, $148$, $332$, and $708$ for $n=3$ to $7$, roughly doubles per qubit [Fig.~\ref{fig:tests}(b)], and for $n=6$ an evolution to $t=2\pi$ with infidelity $2.5\times10^{-3}$ requires $S=64$, i.e.\ about $2\times10^4$ CNOT gates per circuit. The bound $2\sum_{P\neq I}|c_P|=2n$ on $\Delta\lambda$ is loose by a factor $n/2$ for this operator, whereas Eq.~\eqref{eq:gershgorin} is tight ($4$ against a true width of $3.99998$ at $n=10$).

For dense random SPD matrices $A=QQ^\top+N\mathbb{I}$ with $Q_{ij}$ and $b_i$ uniform in $[0,1)$, $L=N(N+1)/2$ and the solutions are delocalized (participation ratio $\approx N/2$). Because $b$ is dense, $\supp(b)$ is the whole space and the rule of Eq.~\eqref{eq:config-set} cannot be applied; configurations are selected by the criterion of each column only. Table~\ref{tab:dense} reports truncation to $R=8$, averaged over $30$ instances per $N$. The spectral widths reach $\approx4100$ at $N=128$, so that the grid of the $4\times4$ example ($t_j=0.25\,j$, $M=8$) lies far outside Eq.~\eqref{eq:nyquist}: $0.25\,\Delta\lambda/\pi=5$--$327$. Lie--Trotter circuits with $S=4$ do not resolve the dynamics on this grid (infidelity $0.995$ for $N=16$ and $0.98$ for $N=32$ at $t=2$), so the table uses the exact distributions $p(z\mid t_j)$. On the grid of the $4\times4$ example, selection by sampling frequency is statistically indistinguishable from random selection, as expected in the dephased regime of Remark~\ref{rem:sampling}; on the Nyquist grid ($t_j=j\pi/\Delta\lambda$, $M=8$, i.e.\ $t_{\max}=8\pi/\Delta\lambda\le0.4$) it reaches $74$--$80\,\%$ of the oracle fidelity but remains below the oracle and the classical greedy selection, which agree to within $0.01$ at every $N$. Single instances vary widely (at $N=32$ the fidelity on the $4\times4$ grid ranges from $0.00$ to $0.45$), so that results for one random matrix are not representative.

\begin{table*}[t]
\caption{Fidelity under truncation to $R=8$ for dense random SPD matrices ($A=QQ^\top+N\mathbb{I}$, $Q_{ij}$ and $b_i$ uniform in $[0,1)$); mean $\pm$ standard deviation over $30$ instances per $N$. ``Exact, $4{\times}4$ grid'': the $8$ configurations with the largest exact $p_{\mathrm{agg}}$ on $t_j=0.25\,j$, $j=1,\dots,8$; ``Exact, Nyquist grid'': the same on $t_j=j\pi/\Delta\lambda$. ``Random'': mean over $200$ random sets per instance. ``Oracle'': the $8$ configurations with the largest $|x_z|$. ``Greedy'': classical residual-based selection started from $\arg\max|b_z|$. $\Delta\lambda$ is the spectral width; $0.25\,\Delta\lambda/\pi$ is the grid step of the $4\times4$ grid in units of the Nyquist step [Eq.~\eqref{eq:nyquist} requires $\lesssim1$].}
\label{tab:dense}
\begin{ruledtabular}
\begin{tabular}{ccccccccc}
$N$ & $L$ & $\Delta\lambda$ & $0.25\,\Delta\lambda/\pi$ & Exact, $4{\times}4$ grid & Exact, Nyquist grid & Random & Oracle & Greedy\\
\hline
16 & 136 & 66 & 5 & $0.46\pm0.13$ & $0.68\pm0.14$ & $0.47\pm0.01$ & $0.88\pm0.05$ & $0.88\pm0.05$\\
32 & 528 & 262 & 21 & $0.23\pm0.13$ & $0.48\pm0.14$ & $0.22\pm0.01$ & $0.61\pm0.07$ & $0.60\pm0.07$\\
64 & 2080 & 1035 & 82 & $0.10\pm0.06$ & $0.25\pm0.08$ & $0.10\pm0.01$ & $0.33\pm0.04$ & $0.34\pm0.03$\\
128 & 8256 & 4108 & 327 & $0.04\pm0.02$ & $0.13\pm0.05$ & $0.05\pm0.00$ & $0.17\pm0.04$ & $0.18\pm0.01$\\
\end{tabular}
\end{ruledtabular}
\end{table*}

\subsection{Scaling with the matrix dimension}
\label{sec:scalability}
Figure~\ref{fig:scalability} shows how the method scales with the matrix dimension $N$ for the banded (screened Poisson) family of Sec.~\ref{sec:localized} and for dense random SPD matrices, both with a point source (unlike the dense $b$ of Table~\ref{tab:dense}). The upper panels show the cost of the quantum step, i.e.\ the number of Pauli terms and CNOT gates per circuit~(a), and the fraction $R/N$ of configurations selected by sampling~(b). The lower panels show the accuracy: the relative error when all sampled configurations are used~(c), and the fidelity after truncation to the $R=\max(2,N/8)$ most frequently sampled configurations, compared with the oracle selection~(d).

For the banded family, $R$ saturates at $\approx33$ for $N\ge64$, so that $R/N$ decreases as $1/N$ while the error remains constant ($\approx10^{-5}$), whereas the circuit cost grows with $N$. For $N\ge32$ the truncated selection by sampling reproduces the oracle to the resolution of the plot, for both the noiseless circuits and the exact distributions; for the smallest subspaces ($R=2$, $N=8$ and $16$), the noisy simulator at $N=8$ and the noiseless simulator at $N=16$ fall below the oracle ($F=0.61$ and $0.60$ against $0.76$ and $0.745$). The dense family could be run with circuits only up to $N=32$, where a circuit contains $\sim10^4$ CNOT gates, and all configurations are sampled up to $N=64$. With exact distributions, $R/N$ drops to $0.875$ and $0.47$ at $N=128$ and $256$, and the error of the untruncated solution rises to a few $10^{-2}$: for delocalized solutions, a smaller subspace comes at the price of accuracy.

\begin{figure*}[t]
\centering
\includegraphics[width=0.92\textwidth]{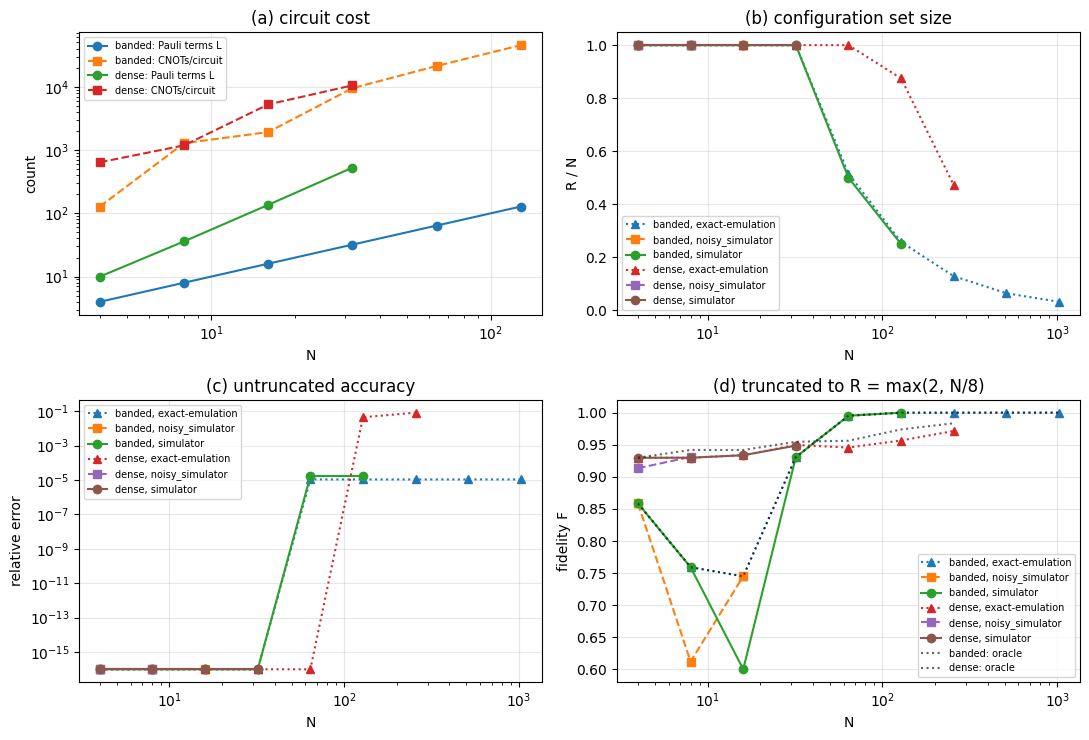}
\caption{Scalability test for the banded (screened Poisson) and dense (random
SPD) families with a point source (unlike the dense $b$ of Table~\ref{tab:dense}), $M=8$, $\Delta t=\pi/\Delta\lambda$
(Gershgorin), and $N_{\mathrm{shots}}=2000$ per time point. Solid lines:
noiseless \texttt{AerSimulator}; dashed: \texttt{AerSimulator} with a noise
model; dotted: sampling of the exact distributions (no circuits).
(a)~Pauli terms $L$ and CNOT gates per circuit, with the number of Trotter
steps chosen for an infidelity below $10^{-2}$ at $t_{\max}$.
(b)~Relative size $R/N$ of the untruncated configuration set.
(c)~Relative error of the untruncated solution (exact results plotted at
$10^{-16}$).
(d)~Fidelity after truncation to $R=\max(2,N/8)$ configurations; dotted black
lines: oracle selection of the $R$ largest $|x_z|$.}
\label{fig:scalability}
\end{figure*}

\subsection{Is the quantum step necessary?}
\label{sec:classical}
For the screened Poisson problem, the untruncated configuration set coincides with the classical Krylov support $\bigcup_{m\le16}\supp(A^mb)$ obtained with $16$ sparse matrix--vector products: both contain $33$ configurations and give the same fidelity to ten digits [Fig.~\ref{fig:tests}(a)]. Classical residual-based selection, which starts from $\supp(b)$ (from $\arg\max|b_z|$ for the dense $b$ of Table~\ref{tab:dense}) and repeatedly adds the configuration with the largest $|(b-Ax_{\mathrm{approx}})_z|$, reproduces the oracle in Tables~\ref{tab:localized} and~\ref{tab:dense}, and the conjugate-gradient method converges to a residual of $10^{-10}$ in $34$ iterations. The same conclusion was reached for elliptic problems with up to $4096$ unknowns, where classical residual-greedy selection outperformed quantum sampling in every case studied~\cite{ahn2026riemannian}. These observations are general: when $L=\mathrm{poly}(n)$ and the solution is concentrated on few configurations, the residual of a Galerkin solution is supported on $\mathcal{C}$ and its neighbors, at most $RL$ configurations, so classical residual-based selection is efficient in the same regime. Moreover, compression requires a small light cone, within which the evolution of a basis state can be simulated classically. A quantum advantage therefore requires problems in which the relevant configurations are not reachable by short classical expansions from $\supp(b)$, or in which sampling from $e^{-iAt}\ket{b}$ is classically hard at the required times.

\section{Discussion and outlook}
\label{sec:outlook}
The results point to the following directions. First, problem classes with genuinely polynomial circuits are required. Shifted local Hamiltonians, $A=H+\sigma\mathbb{I}$ with a $k$-local $H$, have $L=O(n)$ and turn the linear system into a resolvent problem related to Green's functions, for which SKQD has recently been developed; alternatively, banded matrices can be evolved with shift circuits rather than Pauli strings~\cite{patel2026green}. Second, the circuit depth for dense matrices can be reduced by truncating Eq.~\eqref{eq:pauli} to the $L'$ largest coefficients, which is Frobenius-optimal at fixed $L'$ because the Pauli basis is orthogonal. Third, Trotter-induced spurious configurations must be controlled, for instance by count thresholds, at the risk of the truncation errors discussed above, or by randomized compilation such as qDRIFT, for which error bounds are available~\cite{park2026sqls}. Fourth, problems must be identified in which the relevant configurations are not accessible to the classical strategies of Sec.~\ref{sec:classical}, and hardware runs in the truncated regime and comparisons with HHL and classical iterative solvers at matched size are needed. Finally, the heuristic time grid of Eq.~\eqref{eq:nyquist} could be replaced by the sampling schedules with provable shot complexity of Ref.~\onlinecite{park2026sqls}.

\section{Conclusions}
\label{sec:concl}
We have extended sample-based quantum diagonalization to systems of linear equations, complementary to recently proposed sample-based quantum linear solvers~\cite{ahn2026riemannian,park2026sqls}, by combining the SKQD sampling protocol with a Galerkin projection onto the sampled configurations, which always include the support of $\mathbf{b}$. The subspace solution is exact when the sampled configurations contain the Krylov support of $\mathbf{b}$, optimal in the energy norm for symmetric positive-definite matrices, and insensitive to noise-induced configurations as long as they do not displace the ideal ones. On a localized problem with $N=1024$ (with exact sampling distributions), the method selects $33$ configurations, reaches a relative error of $10^{-5}$, and its truncated subspaces are optimal; on IBM Quantum hardware, the sampled distribution of a $4\times4$ example deviates visibly from the exact one, but the solution is recovered exactly because all configurations are sampled. The same tests show that the selected subspace coincides with that of inexpensive classical selection, in agreement with Ref.~\onlinecite{ahn2026riemannian}, that banded matrices require exponentially many Pauli terms in the binary encoding, and that Trotter errors inflate the subspace. With Pauli-string encoding and Lie--Trotter evolution, the circuits are therefore not shallow for banded matrices, the noise robustness has been demonstrated only for readout errors and, with gate errors and on hardware, only for configuration sets that exhaust the space, and the hardware test ($N=4$) does not yet probe the truncated regime. Identifying problems in which the relevant configurations are inaccessible to classical selection, and demonstrating the method on hardware in that regime, is the central open question.

\section*{Code availability}
This article is accompanied by the open-source Python package \texttt{sqd\_linsolve} and two Jupyter
notebooks, provided as ancillary files (directory \texttt{anc/}) and released under the MIT License.
The notebook \texttt{reproduce\_paper.ipynb} uses the package to reproduce the numerical results of
this article and compares each value with the one reported here. The notebook
\texttt{sqd\_linear\_systems\_selfcontained.ipynb} contains a self-contained implementation of the
method with a step-by-step worked example, explanations of each step, and practical information on
running the method on simulators and IBM Quantum hardware. The per-time counts of the IBM Quantum job
of Fig.~\ref{fig:hardware_vs_simulation} are provided in \texttt{anc/data/}; running new hardware jobs requires an IBM Quantum account.

\begin{acknowledgments}
Preliminary results of this work were presented at the Quantum AI 2026 Conference at Indiana University Indianapolis, Indianapolis, IN, USA, 14--16 August 2026. 

We acknowledge the use of IBM Quantum services for this work. The views expressed are those of the author, and do not reflect the official policy or position of IBM or the IBM Quantum team.

\end{acknowledgments}
\clearpage
\newpage
\bibliography{ref}
\bibliographystyle{jcp}
\end{document}